\documentclass[copyright,creativecommons]{eptcs}

\usepackage[titletoc]{appendix}
\usepackage{amsmath}
\usepackage{amssymb}
\usepackage{graphicx}
\usepackage{enumerate}
\usepackage{enumitem}
\usepackage{mathtools}
\usepackage{multicol}
\usepackage{proof}
\usepackage{amsthm}
\usepackage[normalem]{ulem}
\usepackage{array}

\theoremstyle{plain}
\newtheorem{lemma}{Lemma}
\newtheorem{theorem}{Theorem}
\newtheorem{proposition}{Proposition}

\theoremstyle{definition}

\newtheorem{definition}{Definition}

\theoremstyle{remark}
\newtheorem{example}{Example}
\newtheorem{remark}{Remark}

\usepackage{tikz}
\usepackage{tikz-cd}
\usetikzlibrary{patterns,arrows, topaths, calc, positioning}
\tikzset{>=stealth}
\tikzstyle{node} = [circle, minimum size = 1.4mm, inner sep = 0mm, color=black, fill]
\tikzstyle{hyperedge} = [rectangle, minimum width = 5mm, minimum height = 5mm, draw, inner sep = 0mm]
\tikzstyle{hyperedgewide} = [rectangle, minimum width = 8mm, minimum height = 5mm, draw, inner sep = 0mm]
\tikzstyle{HG} = [align = center]
\tikzstyle{circledge} = [circle, minimum size = 7mm, inner sep = 0mm, color=black, draw]

\pgfdeclarepatternformonly{NELines}{\pgfqpoint{-1pt}{-1pt}}{\pgfqpoint{4pt}{4pt}}{\pgfqpoint{3pt}{3pt}}%
{
	\pgfsetlinewidth{0.4pt}
	\pgfpathmoveto{\pgfqpoint{0pt}{0pt}}
	\pgfpathlineto{\pgfqpoint{1pt}{1pt}}
	\pgfpathmoveto{\pgfqpoint{2pt}{2pt}}
	\pgfpathlineto{\pgfqpoint{3.1pt}{3.1pt}}
	\pgfusepath{stroke}
}

\newcommand{\diam}{\Diamond}
\newcommand{\eqdef}{\mathrel{\mathop:}=}
\newcommand{\SG}{\mathop{\mathrm{sg}}}
\newcommand{\LC}{\mathrm{L}}
\newcommand{\NL}{\mathrm{NL}}
\newcommand{\NLM}{\mathrm{NL}\diam}

\newcommand{\I}{\mathbb{I}}
\newcommand{\J}{\mathbb{J}}
\newcommand{\A}{\mathbb{A}}
\newcommand{\C}{\mathbb{C}}
\newcommand{\Dis}{\mathrm{D}}
\newcommand{\HL}{\mathrm{HL}}
\newcommand{\MILLFO}{\mathrm{MILL1}}

\newcommand{\htree}{\mathrm{ht}_\mu}
\newcommand{\htreebar}{\mathrm{ht}_{\bar{\mu}}}

\newcommand{\head}{\mathop{\mathrm{head}}}

\newcommand{\LP}{\mathrm{LP}}
\newcommand{\rk}{\mathit{rk}}

\title{Multimodality in the Hypergraph Lambek Calculus}
\author{
Tikhon Pshenitsyn
\institute{
	Lomonosov Moscow State University, GSP-1, Leninskie Gory, Moscow, 119991, Russia
\\
	Steklov Mathematical Institute of Russian Academy of Sciences, 8 Gubkina St. Moscow,
	119991, Russia
}
\thanks{The study was funded by the Interdisciplinary Scientific and Educational School of Moscow University ``Brain, Cognitive Systems, Artificial Intelligence''; the Theoretical Physics and Mathematics Advancement Foundation ``BASIS''.
}
\email{ptihon at yandex.ru}
}
\def\titlerunning{Multimodality in the Hypergraph Lambek Calculus}
\def\authorrunning{Tikhon Pshenitsyn}

\begin{document}
	\maketitle
\begin{abstract}
	The multimodal Lambek calculus $\NLM$ is an extension of the Lambek calculus that includes several product operations (some of them being commutative or/and associative), unary modalities, and corresponding residual implications. In this work, we relate this calculus to the hypergraph Lambek calculus $\HL$. The latter is a general pure logic of residuation defined in a sequent form; antecedents of its sequents are hypergraphs, and the rules of $\HL$ involve hypergraph transformation.
	Our main result is the embedding of the multimodal Lambek calculus (with at most one associative product) in $\HL$. It justifies that $\HL$ is a very general Lambek-style logic and also provides a novel syntactic interface for $\NLM$: antecedents of sequents of $\NLM$ are represented as tree-like hypergraphs in $\HL$, and they are derived from each other by means of hyperedge replacement. The advantage of this embedding is that commutativity and associativity are incorporated in the sequent structure rather than added as separate rules. Besides, modalities of $\NLM$ are represented in $\HL$ using the product and the division of $\HL$, which explicitizes their residual nature.
\end{abstract}

\section{Introduction}\label{sec_introduction}

The Lambek calculus $\LC$ introduced in \cite{Lambek58} is a logic developed to model syntax of natural languages. It has the non-commutative product operation $\bullet$ and two adjoints $\backslash, /$ satisfying the residuation law: $A \bullet B \to C$ is derivable in $\LC$ if and only if so is $A \to C/B$ if and only if so is $B \to A \backslash C$. Despite being non-commutative, the product is associative, i.e. $(A \bullet B) \bullet C \leftrightarrow A \bullet (B \bullet C)$ is derivable in $\LC$. One can either accept or reject associativity and commutativity of $\bullet$, which results in four logics: $\LC$, $\NL$ (non-associative, non-commutative), $\LP$ (associative and commutative), $\mathrm{NLP}$ (non-associative, commutative).

For linguistic purposes, one would like to have control over associativity and commutativity. To do this, one can introduce several modes by considering several products $\bullet_i$ for $i \in \I$ ($\I$ is the set of modes) with divisions corresponding to them and then define structural rules for each of the operations independently. Besides, one can also add unary operations, i.e. modalities, to $\LC$. The residuation law, which ``lies at the heart of categorial type logic'' \cite{Moortgat96}, is adapted for modalities as follows: a modality $\diam$ has a pair $\square$, and $\diam A \to B$ is derivable if and only if $A \to \square B$ is derivable (this means that $\diam$ and $\square$ form Galois connection). The Lambek calculus with modes and unary operations, which is called \emph{the multimodal Lambek calculus $\NLM$}, is introduced in \cite{Moortgat96}; its overview can also be found in \cite{MootR12}. In general, a family of modalities $\square_j$, $\diam_j$ for $j \in \J$ is considered (assume that $\I \cap \J = \emptyset$). 

The calculus $\NLM$ is defined as follows (in a sequent way): if $\mathcal{F}(\NLM)$ is the set of formulas of $\NLM$, which are built using $\bullet_i,\backslash_i,/_i,\square_j,\diam_j$, then \emph{the set of structured databases} $\mathcal{T}$ is defined as $\mathcal{T} \eqdef \mathcal{F}(\NLM) \mid (\mathcal{T}, \mathcal{T})^i \mid \langle\mathcal{T}\rangle^j$ where $i \in \I$, $j \in \J$; in other words, there are several kinds of brackets corresponding to products and several kinds of brackets corresponding to modalities. A sequent is of the form $\Pi \to A$ where $\Pi$ is a structured database and $A$ is a formula. The rules of $\NLM$ are the following:

$$
\infer[(ax)]
{
	A \to A
}
{
}
\quad
\infer[(\bullet_i L)]
{
	\Delta[A \bullet_i B] \to C
}
{
	\Delta[(A, B)^i] \to C
}
\quad
\infer[(\bullet_i R)]
{
	(\Pi_1, \Pi_2)^i \to A_1 \bullet_i A_2
}
{
	\Pi_1 \to A_1
	&
	\Pi_2 \to A_2
}
$$
$$
\infer[(\backslash_i L)]
{
	\Delta[(\Pi, B \mathbin{\backslash_i} A)^i] \to C
}
{
	\Delta[A] \to C
	&
	\Pi \to B
}
\quad
\infer[(\backslash_i R)]
{
	\Pi \to B \mathbin{\backslash_i} A
}
{
	(B, \Pi)^i \to A
}
\quad
\infer[(/_i L)]
{
	\Delta[(A \mathbin{/_i} B, \Pi)^i] \to C
}
{
	\Delta[A] \to C
	&
	\Pi \to B
}
\quad
\infer[(/_i R)]
{
	\Pi \to A \mathbin{/_i} B
}
{
	(\Pi, B)^i \to A
}
\quad
$$
$$
\infer[(\diam_j L)]
{
	\Delta[\diam_j A] \to C
}
{
	\Delta[\langle A \rangle^j] \to C
}
\quad
\infer[(\diam_j R)]
{
	\langle \Pi \rangle^j \to \diam_j(C)
}
{
	\Pi \to C
}
\quad
\infer[(\square_j L)]
{
	\Delta[\langle \square_j A\rangle^j] \to C
}
{
	\Delta[A] \to C
}
\quad
\infer[(\square_j R)]
{
	\Pi \to \square_j C
}
{
	\langle  \Pi \rangle^j \to C
}
$$
This logic can be called \emph{the pure logic of residuation} for $\bullet_i,\backslash_i,/_i,\square_j,\diam_j$ \cite{Moortgat96}. In general, one would like to distinguish a subset $\C \subseteq \I$ and $\A \subseteq \I$ and add the following postulates for $c \in \C$ and $a \in \A$:
$$
\infer[(P)]
{
	\Delta[(\Pi_2, \Pi_1)^c] \to C
}
{
	\Delta[(\Pi_1, \Pi_2)^c] \to C
}
\quad
\infer[(A_l)]
{
	\Delta[((\Pi_1, \Pi_2)^a,\Pi_3)^a] \to C
}
{
	\Delta[(\Pi_1, (\Pi_2,\Pi_3)^a)^a] \to C
}
\quad
\infer[(A_r)]
{
	\Delta[(\Pi_1, (\Pi_2,\Pi_3)^a)^a] \to C
}
{
	\Delta[((\Pi_1, \Pi_2)^a,\Pi_3)^a] \to C
}
$$
Apart from them, one can introduce non-linear rules like weakening or contraction. It is also possible to add various structural postulates for modalities as well as interaction principles for operations, which results in a huge family of logics. Our attention, however, will be devoted only to the above rules. 

There are many modifications of the Lambek calculus apart from $\NLM$. Most of them can be introduced in a sequent form, which is convenient for proof search and proof analysis. These modifications have different sets of operations (compare e.g. $\LC$, $\NLM$ and the displacement calculus \cite{MorrillVF11}), and sequents have different forms. One might be interested in whether these different calculi can be treated uniformly using some general formalism. In particular, one would expect that the products $\bullet_i$ and the modalities $\diam_j$ are treated in this formalism as particular cases of the same operation as well as the divisions and the $\square_j$ modalities (since $\diam_j$ is a truncated product and $\square_j$ is its residual implication as noted in \cite{Moortgat96}).

One of general logics of pure residuation is the first-order multiplicative intuitionistic linear logic $\MILLFO$ \cite{MootP01}. In \cite{Moot14, MootP01}, it is shown that $\LC$, $\LP$, $\NL$, and the displacement calculus $\Dis$ (defined in \cite{MorrillVF11}) can be embedded in $\MILLFO$. However, in the section of \cite{MootP01} devoted to the Lambek calculus with modalities $\square$, $\diam$, the authors say that ``it is unclear [...] if it is possible to give an embedding translation for these connectives''. After that, no work on finding this translation has been published, to our best knowledge. Besides, the authors define an embedding of $\mathrm{NLP}$ in $\MILLFO$ but then they say that an additional rule must be added to $\mathrm{NLP}$ in order for the translation to be correct.

Another generalized logic of pure residuation is the hypergraph Lambek calculus $\HL$ defined in \cite{Pshenitsyn22}. It is defined as a sequent calculus that operates on hypergraph formulas and sequents using the hyperedge replacement operation \cite{DrewesKH97}. Sequents of $\HL$ are of the form $H \to C$ where $C$ is a formula of $\HL$ and $H$ is any hypergraph, whose hyperedges are labeled by formulas of $\HL$. Formulas are built from primitive ones using the hypergraph product operation $\times$ and the hypergraph division operation $\div$. The inductive definition of $\HL$ says e.g. that if $M$ is a hypergraph labeled by formulas, then $\times(M)$ is a formula. 

In \cite{Pshenitsyn22}, it is proved that the Lambek calculus can be embedded in $\HL$; in this embedding, antecedents of sequents of $\LC$ are represented as string graphs in $\HL$. It is also shown there how to embed $\LP$ in $\HL$ by using hypergraphs of the star shape where any permutation of edges leads to an isomorphic graph. 

\emph{The main result of this paper} is the embedding of $\NLM$ in $\HL$ if there is at most one associative mode ($|\A| \le 1$). The main merit of this embedding can be explained on the following example. If we are given the Lambek calculus, we can consider it as $\NLM$ without modalities and with one mode $\alpha$, for which the associativity rules $(A_l)$, $(A_r)$ are added. However, $\LC$ can also be defined by changing the structure of sequents: namely, let us say that a sequent is of the form $A_1,\dotsc,A_n \to B$ where the antecedent is a string of types instead of a structured database. Then, associativity is incorporated in the sequent structure, so no rules for it are required. Similarly, the logic $\LP$ can be obtained from $\LC$ by adding the permutation rule $(P)$; however, instead one can define antecedents of sequents as multisets, thus making the permutation rule a part of the sequent. Our embedding does the same in a much more general case: given any sets $\I$, $\J$, $\C$, $\A$ (with $|\A|\le 1$), we represent antecedents of sequents of $\NLM$ in $\HL$ as certain kinds of hypergraphs that incorporate all the structural postulates for the products. Such hypergraphs are similar to hypertrees considered in \cite{Moot08}, however, they are more general than those from the cited article.

In Section \ref{sec_preliminaries}, we introduce preliminary notions concerning hypergraphs and hyperedge replacement. In Section \ref{ssec_HL_def}, we introduce the definition of $\HL$ and explain it following the concept of residuation; then, in Section \ref{ssec_NLM_in_HL}, we present an embedding of $\NLM$ in $\HL$. In Section \ref{sec_discussion_conclusion}, we give concluding remarks.


\section{Preliminaries}\label{sec_preliminaries}
The set $\Sigma^\ast$ is the set of strings over an alphabet $\Sigma$ including the empty word $\Lambda$. By $w(i)$ we denote the $i$-th symbol of $w \in \Sigma^\ast$; $|w|$ is the number of symbols in $w$. Hereinafter, a string can be represented either as $a_1\dotsc a_n$ or as $a_1,\dotsc,a_n$ (using commas as separators between symbols). Each function $f:\Sigma \to \Delta$ can be extended to a homomorphism $f:\Sigma^\ast \to \Delta^\ast$. $[n]$ denotes the set $\{1,2,\dotsc,n\}$, and $[0] \eqdef \emptyset$.

A \emph{ranked set} $M$ is the set $M$ along with a function $\rk:M \to \mathbb{N}$ called \emph{the rank function}. Given a ranked set of \emph{labels} $\Sigma$, a \emph{hypergraph} $G$ over $\Sigma$ is a tuple $\langle V_G, E_G, att_G, lab_G, ext_G \rangle$ where $V_G$ is a finite set of \emph{nodes}, $E_G$ is a finite set of \emph{hyperedges}, $att_G: E_G\to V_G^\ast$ assigns a string of \emph{attachment nodes} to each hyperedge, $lab_G: E_G \to \Sigma$ labels each hyperedge by some element of $\Sigma$ in such a way that $\rk(lab_G(e))=|att_G(e)|$ whenever $e\in E_G$, and $ext_G\in V_G^\ast$ is a string of \emph{external nodes}.  Hypergraphs are considered up to isomorphism \cite{DrewesKH97}. The set of all hypergraphs with labels from $\Sigma$ is denoted by $\mathcal{H}(\Sigma)$. The rank function $\rk$ is defined on $\mathcal{H}(\Sigma)$ as follows: $\rk_G(e)\eqdef|att_G(e)|$. Besides, let $\rk(G)\eqdef |ext_G|$ be the rank function on hypergraphs. Thus $\mathcal{H}(\Sigma)$ is also a ranked set.

In drawings of hypergraphs, small circles are nodes, and rectangles are hyperedges. To represent attachment nodes $att_G(e)$, we draw a line with number $i$ from $e$ to $v$, if $v$ is the $i$-th symbol of $att_G(e)$. The label $lab_G(e)$ is put on the hyperedge $e$. External nodes are represented by numbers in round brackets: if the $i$-th symbol of $ext_G$ is $v$, then we mark it as $(i)$. If a hyperedge has exactly two attachment nodes, then it is depicted by a labeled thick arrow that goes from the first attachment node to the second one.

The \emph{replacement of a hyperedge $e_0$ in $G$ ($e_0 \in E_G$) by a hypergraph $H$} (such that $rk(e_0)=rk(H)$) is the following hypergraph $K$ denoted as $G[e_0/H]$:
\begin{itemize}
	\item $V_K = (V_G \sqcup V_H)/\equiv$ where $\equiv$ is the smallest equivalence relation such that $att_G(e_0)(i)\equiv ext_H(i)$ for $i = 1,\dotsc, \rk(e_0)$;
	\item $E_K = \left(E_G \setminus \{e_0\}\right) \sqcup E_H$;
	\item $att_K(e) = [att_G(e)]_\equiv$ for $e \in E_G$; $att_K(e) = [att_H(e)]_\equiv$ for $e \in E_H$ (here $[v]_\equiv$ is the equivalence class of $v$, and $[v_1\dotsc v_n]_\equiv$ is defined as a homomorphic extension $[v_1]_\equiv\dotsc [v_n]_\equiv$);
	\item $lab_K(e) = lab_G(e)$ for $e \in E_G$; $att_K(e) = lab_H(e)$ for $e \in E_H$;
	\item $ext_K = [ext_G]_\equiv$.
\end{itemize}
If several hyperedges of a hypergraph are replaced by other hypergraphs, then the result does not depend on the order of the replacements; moreover the result is the same if replacements are done simultaneously \cite{DrewesKH97}. If $e_1,\dots,e_k$ are distinct hyperedges of a hypergraph $G$ and they are simultaneously replaced by hypergraphs $H_1,\dots,H_k$ resp., then the result is denoted by $G[e_1/H_1,\dots,e_k/H_k]$.

Let us now define some auxiliary notions. Given $H \in \mathcal{H}(\Sigma)$ and a function $f:E_H\to \Sigma$ such that $rk(e)=rk(f(e))$ for $e\in E_H$, \emph{a relabeling $f(H)$} is the hypergraph $f(H)=\langle V_H, E_H, att_H, f, ext_H\rangle$. If $H$ is a hypergraph and $e \in E_H$, then $H-e$ is obtained by removing $e$ from $E_H$ and restricting $att_H$ and $lab_H$. 

A \emph{handle} $a^\bullet$ is a hypergraph $\langle [n],[1],att,lab,1\dotsc n\rangle$ where $att(1)=1\dotsc n$ and $lab(1)=a$ ($a\in\Sigma$, $rk(a)=n$). The handle is in some sense a ``neutral'' hypergraph. In particular, if for some hypergraph $G$ and $e \in E_G$ the replacement $G[e/a^\bullet]$ is correct, then it is simply a relabeling of the hyperedge $e$ by $a$.

\emph{A string graph $\SG(w)$ induced by a string $w=a_1\dots a_n$} is the hypergraph with $V_{\SG(w)} = \{v_0,\dotsc,v_n\}$, $E_{\SG(w)} = \{s_1,\dotsc,s_n\}$, $att_{\SG(w)}(s_i)=v_{i-1}v_i$, $lab_{\SG(w)}(s_i)=a_i$, $ext_{\SG(w)}=v_0v_n$. This is the canonical way of representing a string as a hypergraph in the framework we work in; in particular, context-free grammars are translated into hyperedge replacement grammars using string graphs \cite{DrewesKH97}. Let us denote the graph $\SG(XY)$ where $X,Y$ are symbols as $Str$ (we will use this notation when we do not care about the actual labels of $Str$). The following property of $Str=\SG(XY)$ is important for further considerations:
\begin{proposition}\label{proposition_SG}
	If $\alpha_1$, $\alpha_2$ are strings, then $Str[s_1/\SG(\alpha_1),s_2/\SG(\alpha_2)] = \SG(\alpha_1 \alpha_2)$.
\end{proposition}

\section{How The Hypergraph Lambek Calculus Models Modes and Modalities}\label{sec_main}

\subsection{Hypergraph Lambek Calculus as a General Theory of Residuation}\label{ssec_HL_def}

Our goal is to show that the hypergraph Lambek calculus $\HL$ studied in \cite{Pshenitsyn21,Pshenitsyn22,Pshenitsyn23} generalizes the multimodal Lambek calculus $\NLM$. In order to reach it, firstly, we would like to discuss how $\HL$ works. A detailed explanation of $\HL$ can be found in \cite{Pshenitsyn21}: there $\HL$ is motivated, on the one hand, by a conversion procedure of context-free grammars into Lambek categorial grammars, and, on the other hand, by generalizing the language semantics of $\LC$ to hypergraphs. In this section, we show that the definition of $\HL$ can also be approached if one is guided by the following postulates:
\begin{enumerate}
	\item\label{postulates_HL_0_goal} The logic $\HL$ that should play the role of a very general Lambek-style calculus deals with hypergraphs as with the most general finite discrete structures.
	\item\label{postulates_HL_1_plot} As a sequent calculus, $\HL$ must be defined by specifying the set of formulas $\mathcal{F}(\HL)$, the notion of a sequent and by introducing axioms and rules.
	\item\label{postulates_HL_2_type} There must be a generalized product operation in the inventory of operations of $\HL$, which unifies various ways of composition in existing Lambek-style calculi; there must also be a corresponding division operation, which is related to the generalized product by the residuation law. 
	\item\label{postulates_HL_3_sequent} A sequent of $\HL$ must be of the form $H \to C$ where $C \in \mathcal{F}(\HL)$ and $H \in \mathcal{H}(\mathcal{F}(\HL))$, i.e. where $H$ is a hypergraph labeled by formulas of the calculus.
	\item\label{postulates_HL_3.5_rules} Rules of $\HL$ must be formulated using the notion of hyperedge replacement.
	\item\label{postulates_HL_4_translation} There must exist a function $tr:\mathcal{F}(\LC) \to \mathcal{F}(\HL)$ such that a sequent $A_1,\dotsc,A_n \to B$ is derivable in $\LC$ if and only if $\SG(tr(A_1)\dotsc tr(A_n)) \to tr(B)$ is derivable in $\HL$. In other words, there must be a translation $tr$ of formulas that is extended to sequents as follows: an antecedent of a sequent of $\LC$ must be translated into a string graph in $\HL$. 
	\item\label{postulates_HL_5_translation_rules} Even more, if there is an instance of a rule of $\LC$ and all sequents contained in it are translated into those of $\HL$ in a way described in Postulate \ref{postulates_HL_4_translation}, then the result must be an instance of a rule of $\HL$.
\end{enumerate}
Postulate \ref{postulates_HL_2_type} is brought to life in the hypergraph Lambek calculus as follows: hypergraphs can be considered as mechanisms of composing, or concatenating structures. Look at Proposition \ref{proposition_SG}: there it is stated that the replacement of hyperedges $s_1$, $s_2$ in $Str$ by string graphs $\SG(\alpha_1)$, $\SG(\alpha_2)$ leads to a string graph corresponding to concatenation of $\alpha_1$ and $\alpha_2$. Consequently, the product $p \bullet q$ of $\LC$ can be represented by the string graph $\SG(pq)$. Generalizing this idea, we can say that, given types $p$ and $q$, a hypergraph $M$ with two hyperedges $m_1$, $m_2$ such that $lab_M(m_1) = p$ and $lab_M(m_2)=q$ is a composition of $p$ and $q$; the structure of $M$ is understood as the way of composing resources $p$ and $q$.  The author of \cite{Pshenitsyn21} introduces a symbol $\times$ and says that, if $M$ is any hypergraph labeled by formulas of $\HL$, then $\times(M)$ must also be a formula of $\HL$, which is understood as a composition of formulas that are labels of $M$ into a single structure $M$. In particular, the formula $p \bullet q$ of $\LC$ is translated into the formula $\times(\SG(pq))$ of $\HL$.

To come up with the notion of a generalized division, let us revise the residuation law following \cite{Moortgat96}. There, given an $n$-ary product $f_\bullet(A_1,\dotsc,A_n)$, one introduces the $i$-th place residual $f^i_\to(A_1,\dotsc,A_n)$; then the residuation law is as follows: $f_\bullet(A_1,\dotsc,A_n) \to B$ if and only if $A_i \to f^i_\to(A_1,\dotsc, A_{i-1},B,A_{i+1},\dotsc,A_n)$. In $\HL$, instead of $f_\bullet$ we have the formula $\times(M)$ as the product, and instead of ``the $i$-th place'' we have hyperedges of $M$ as places for formulas. Let $m \in E_M$ be a hyperedge with the label $lab_M(m)=A$; then we would like to say that $\times(M) \to C$ if and only if $A \to Res(C,M-m)$ where $Res$ is a hypothetical residual operation involving $C$ and $M$ without the distinguished hyperedge $m$. However, we cannot simply remove $m$: we have to remember the position where it is attached. To do this, one introduces a set of labels $\{\$_n\}_{n \in \mathbb{N}}$ ($\rk(\$_n)=n$), which can be understood as ``holes''; a hole label is designed to be placed on $m$ instead of $A$. Then, the residuation law can be formulated as follows: $\times(M) \to C$ if and only if $A \to Res(C,M^\prime)$ where $M^\prime = M[m/\$_n^\bullet]$ ($n=\rk(m)$) is a hypergraph with the hole instead of the hyperedge $m$. The function $Res(C,M^\prime)$ is denoted in \cite{Pshenitsyn23} as $C \div M^\prime$.

The rigorous definitions formalizing the above discussion are presented in \cite{Pshenitsyn23} in the following way:


\begin{definition}[\cite{Pshenitsyn23}]\label{definition_HL_formula}
	Let us fix a ranked set $\mathit{Pr}$ called \emph{the set of primitive formulas} such that for each $k \in\mathbb{N}$ there are infinitely many $p \in \mathit{Pr}$ satisfying $\rk(p) = k$. Besides, let us fix a countable set of labels $\$_n,n\in\mathbb{N}$ and set $\rk(\$_n)=n$; let us agree that these labels do not belong to any other set considered in the definition of the calculus. Then the ranked set of \emph{formulas} $\mathcal{F}(\HL)$ (called \emph{types} in \cite{Pshenitsyn23}) is the least set satisfying the following conditions:
	\begin{enumerate}
		\item All primitive formulas are formulas.
		\item Let $N \in \mathcal{F}(\HL)$ be a formula, and let $D$ be a hypergraph such that labels of all its hyperedges, except for one, are from $\mathcal{F}(\HL)$, while the remaining one is $\$_d$ for some $d$; let also $\rk(N)=\rk(D)$. Then $N\div D \in \mathcal{F}(\HL)$ and $\rk(N\div D) = d$. The hyperedge of $D$ labeled by $\$_d$ is denoted by $e^\$_D$.
		\item If $M$ is a hypergraph with labels from $\mathcal{F}(\HL)$, then $\times(M) \in \mathcal{F}(\HL)$ and $\rk(\times(M)) = \rk(M)$.
	\end{enumerate}
\end{definition} 
\begin{definition}[\cite{Pshenitsyn23}]\label{definition_HL_sequent}
	A \emph{sequent} is a structure of the form $H\to A$ where $H \in \mathcal{H}(\mathcal{F}(\HL))$ is the \emph{antecedent} of the sequent, and $A \in \mathcal{F}(\HL)$ is the \emph{succedent} of the sequent such that $\rk(H)=\rk(A)$.
\end{definition}

After defining formulas and sequents, let us look at how the Lambek calculus is translated into $\HL$ in \cite{Pshenitsyn22}. The translation function $tr:\mathcal{F}(\LC) \to \mathcal{F}(\HL)$ of $\LC$ in $\HL$ is defined in \cite[Section 3.3]{Pshenitsyn22} as follows:
\begin{multicols}{2}
	\begin{itemize}
		\item $tr(p)\eqdef p$ where $p\in Pr$, $\rk(p)=2$;
		\item $tr(A\bullet B)\eqdef \times(\SG(tr(A),tr(B)))$;
		\item $tr(A/B)\eqdef tr(A)\div \SG(\$_2,tr(B))$;
		\item $tr(B\backslash A)\eqdef tr(A)\div \SG(tr(B),\$_2)$.
	\end{itemize}
\end{multicols}
\noindent
Here the formula $p \bullet q$ is converted into the formula $tr(p \bullet q)$ with the string graph $\SG(pq)$, as expected. 

After defining the notion of formulas and sequents it remains to specify rules of the hypothetical calculus $\HL$, following Postulates \ref{postulates_HL_4_translation} and \ref{postulates_HL_5_translation_rules}. To do this, let us consider the rules of $\LC$ and translate antecedents of sequents participating in them into string graphs according to Postulate \ref{postulates_HL_4_translation}:
\begin{center}
	\begin{tabular}{ccc}
		\vspace{3mm}
		$
		\infer[(\bullet L)]
		{
			\Gamma, A\bullet B, \Delta \to C
		}
		{
			\Gamma, A, B, \Delta \to C
		}
		$
		&
		$\rightsquigarrow$
		&
		$
		\infer[]
		{
			\SG(tr(\Gamma), tr(A\bullet B), tr(\Delta)) \to tr(C)
		}
		{
			\SG(tr(\Gamma), tr(A), tr(B), tr(\Delta)) \to tr(C)
		}
		$
		\\
		$
		\infer[(\bullet R)]
		{
			\Pi_1, \Pi_2 \to A_1 \bullet A_2
		}
		{
			\Pi_1 \to A_1
			&
			\Pi_2 \to A_2
		}
		$
		&
		$\rightsquigarrow$
		&
		$
		\infer[]
		{
			\SG(tr(\Pi_1), tr(\Pi_2)) \to tr(A_1 \bullet A_2) 
		}
		{
			\SG(tr(\Pi_1)) \to tr(A_1)
			&
			\SG(tr(\Pi_2)) \to tr(A_2)
		}
		$
		\\
	\end{tabular}
\end{center}
\begin{example}
	Let $U_1 = tr(p \bullet q)=\times(\SG(p,q))$, $U_2 = tr(r \bullet s)=\times(\SG(r,s))$, $U_3 = tr((p\bullet q)\bullet (r\bullet s)) = \times(\SG(U_2,U_3))$. Below we present an application of $(\bullet R)$ and translate all its sequents into string graphs:
	$$
	\infer[(\bullet R)]
	{
		p,q,r,s \to (p\bullet q)\bullet (r\bullet s)
	}
	{
		p, q \to p \bullet q
		&
		r, s \to r \bullet s
	}
	\quad
	\rightsquigarrow
	\quad
	\infer[]{
		\vcenter{\hbox{{\tikz[baseline=.1ex]{
						\node[node,label=left:{\scriptsize $(1)$}] (N1) {};
						\node[node,right=8mm of N1] (N2) {};
						\node[node,right=8mm of N2] (N3) {};
						\node[node,right=8mm of N3] (N4) {};
						\node[node,right=8mm of N4,label=right:{\scriptsize $(2)$}] (N5) {};
						\draw[->,black] (N1) -- node[above] {$p$} (N2);
						\draw[->,black] (N2) -- node[above] {$q$} (N3);
						\draw[->,black] (N3) -- node[above] {$r$} (N4);
						\draw[->,black] (N4) -- node[above] {$s$} (N5);
		}}}}\to \times(\SG(U_2,U_3))
	}{
		\vcenter{\hbox{{\tikz[baseline=.1ex]{
						\node[node,label=left:{\scriptsize $(1)$}] (N1) {};
						\node[node,right=8mm of N1] (N2) {};
						\node[node,right=8mm of N2,label=right:{\scriptsize $(2)$}] (N3) {};
						\draw[->,black] (N1) -- node[above] {$p$} (N2);
						\draw[->,black] (N2) -- node[above] {$q$} (N3);
		}}}}\to U_2
		&
		\vcenter{\hbox{{\tikz[baseline=.1ex]{
						\node[node,label=left:{\scriptsize $(1)$}] (N1) {};
						\node[node,right=8mm of N1] (N2) {};
						\node[node,right=8mm of N2,label=right:{\scriptsize $(2)$}] (N3) {};
						\draw[->,black] (N1) -- node[above] {$r$} (N2);
						\draw[->,black] (N2) -- node[above] {$s$} (N3);
		}}}}\to U_3
	}
	$$
	The conclusion antecedent $\SG(pqrs)$ can be expressed in terms of the premise antecedents $\SG(pq)$ and $\SG(rs)$ and of the formula $\times(\SG(U_1, U_2))$ as follows: $\SG(pqrs) = \SG(U_1, U_2)[s_1/\SG(pq),s_2/\SG(rs)]$ (see Proposition \ref{proposition_SG}). In other words, we compose the antecedents of the premises using $M=\SG(U_1, U_2)$ as a way of composition and obtain the antecedent of the conclusion. This can be generalized as follows:
	$$
	\infer[(\times R)]
	{
		M[m_1/H_1,\dotsc,m_l/H_l]\to\times(M)
	}
	{
		H_1\to lab_M(m_1) & \dots & H_l\to lab_M(m_l)
	}
	$$
	In other words, if there are sequents with antecedents $H_1$, $\dotsc$, $H_l$ and with succedents being labels of $M$, then one composes $H_1,\dotsc,H_l$ ``according to $M$'' and obtains the sequent $M[m_1/H_1,\dotsc,m_l/H_l]\to\times(M)$ as the conclusion. And, indeed, the rule $(\times R)$ is one of the rules of $\HL$ as defined in \cite{Pshenitsyn23}.
\end{example}
\begin{example}
	Let $U_4 = tr(p \bullet (q \bullet r))$. Below we present a rule application of $(\bullet L)$ and translation of all its sequents into string graphs:
	$$
	\infer[(\bullet L)]
	{
		p\bullet q, r \to (p\bullet q)\bullet r
	}
	{
		p, q, r \to (p\bullet q)\bullet r
	}
	\qquad
	\rightsquigarrow
	\qquad
	\infer[]{
		\vcenter{\hbox{{\tikz[baseline=.1ex]{
					\node[node,label=left:{\scriptsize $(1)$}] (N2) {};
					\node[node,right=17.3mm of N2] (N4) {};
					\node[node,right=8mm of N4,label=right:{\scriptsize $(2)$}] (N5) {};
					\draw[->,black] (N2) -- node[above] {$\times(\SG(pq))$} (N4);
					\draw[->,black] (N4) -- node[above] {$r$} (N5);	
		}}}}\to U_4
	}{
		\vcenter{\hbox{{\tikz[baseline=.1ex]{
					\node[node,label=left:{\scriptsize $(1)$}] (N1) {};
					\node[node,right=8mm of N1] (N2) {};
					\node[node,right=8mm of N2] (N3) {};
					\node[node,right=8mm of N3,label=right:{\scriptsize $(2)$}] (N4) {};
					\draw[->,black] (N1) -- node[above] {$p$} (N2);
					\draw[->,black] (N2) -- node[above] {$q$} (N3);
					\draw[->,black] (N3) -- node[above] {$r$} (N4);
		}}}}\to U_4
	}
	$$
	The premise antecedent $\SG(p,q,r)$ can be expressed in terms of the conclusion antecedent $\SG(tr(p\bullet q), r)$ and the label $\times(\SG(p,q))$ of one of its hyperedges as follows: $\SG(p,q,r) = \SG(tr(p\bullet q), r)[s_1/\SG(p,q)]$. In other words, we replace the hyperedge of the conclusion antecedent with the label of the form $\times(M)$ by $M$ and obtain the antecedent of the premise. In general, this can be formulated as follows:
	$$
	\infer[(\times L)]
	{
		H[e/(\times(M))^\bullet]\to A
	}
	{
		H[e/M]\to A
	}
	$$
	This is the second rule for the product in $\HL$ as defined in \cite{Pshenitsyn23}.
\end{example}

The rules for the hypergraph division can be examined in the same way (unfortunately, we do not do this due to the lack of space). The resulting formal definition is given below according to \cite{Pshenitsyn23}.

\begin{definition}[\cite{Pshenitsyn23}]\label{definition_HL_rules}
	\emph{The axiom and rules of $\HL$} are as follows:
	$$
	\infer[(ax)]
	{
		A^\bullet\to A
	}
	{
	}
	\qquad
	\infer[(\times L)]
	{
		H[e/(\times(M))^\bullet]\to A
	}
	{
		H[e/M]\to A
	}
	\qquad
	\infer[(\times R)]
	{
		M[m_1/H_1,\dotsc,m_l/H_l]\to\times(M)
	}
	{
		H_1\to lab_M(m_1) & \dots & H_l\to lab_M(m_l)
	}
	$$
	$$
	\infer[(\div L)]
	{
		H\left[e/D[e^\$_D/ (N\div D)^\bullet,d_1/H_1,\dotsc,d_k/H_k]\right]\to A
	}{
		H[e/N^\bullet]\to A
		&
		H_1\to lab_D(d_1)&
		\dotsc
		&
		H_k\to lab_D(d_k)
	}
	\qquad\qquad
	\infer[(\div R)]
	{
		F\to N\div D
	}
	{
		D[e^\$_D/F]\to N
	}
	$$
	Here $A$, $N\div D$, $\times(M) \in \mathcal{F}(\HL)$ are formulas of $\HL$; hypergraphs $F,H, H_i \in \mathcal{H}(\mathcal{F}(\HL))$ are labeled by formulas; $e \in E_H$; $E_D=\{e^\$_D,d_1,\dots,d_k\}$, $E_M = \{m_1,\dotsc, m_l\}$. The formula $N \div D$ in the rules $(\div L)$, $(\div R)$ and the formula $\times(M)$ in the rules $(\times L)$, $(\times R)$ are called \emph{major}.
\end{definition}
Definitions \ref{definition_HL_formula}, \ref{definition_HL_sequent}, \ref{definition_HL_rules} constitute the complete description of the hypergraph Lambek calculus $\HL$. 

In \cite{Pshenitsyn22}, it is proved that the residuation law formulated in the following form holds for $\HL$:
\begin{equation}
	\HL \vdash F \to N \div D 
	\mbox{ if and only if }
	\HL \vdash D[e^\$_D/F]\to N \mbox{ for any $F \in \mathcal{H}(\mathcal{F}(\HL))$ and $A \in \mathcal{F}(\HL)$.}
\end{equation}
This justifies that $\HL$ is a pure logic of residuation (although ``pureness'' must be checked independently).

For further reasonings let us introduce some auxiliary notions related to $\HL$ (taken from \cite{Pshenitsyn22}).

\begin{definition}
	Let $\mathcal{F}$ be a subset of $\mathcal{F}(\HL)$. We say that $H\to A$ \emph{is over $\mathcal{F}$} if $G\in\mathcal{H}(\mathcal{F})$ and $A\in\mathcal{F}$. 
\end{definition}
\begin{definition}
	A formula $B$ is \emph{a subformula of $A$} if either $A=B$, or $A=N\div D$ and $B$ is a subformula of $N$ or is a subformula of one of labels of hyperedges from $E_D$, or $A=\times(M)$ and $B$ is a subformula of one of labels of hyperedges from $E_M$. 
\end{definition}
\begin{definition}
	A formula $A$ without subformulas of the form $\times(M)$ where $E_M = \emptyset$ is called \emph{skeleton-free}.
\end{definition}
\begin{definition}
	The \emph{head} of a formula is the set defined inductively as follows:
	\begin{enumerate}
		\item $\head(p) = \{p\}$ for $p \in \mathit{Pr}$;
		\item $\head(N \div D) = \head(N)$;
		\item $\head(\times(M)) = \bigcup_{i=1}^k \head(lab(m_i))$ where $E_M = \{m_1,\dotsc, m_k\}$.
	\end{enumerate}
\end{definition}
The following lemma is proved in \cite{Pshenitsyn22} by a simple induction:
\begin{lemma}\label{lemma_wolf}
	Let $\HL \vdash H \to p$ (where $p \in \mathit{Pr}$); it is given that for each $e \in E_H$ the formula $A=lab_H(e)$ is skeleton-free and that $\head(B)=\{p\}$ implies $B=p$ for each subformula $B$ of $A$. Then $H=p^\bullet$.
\end{lemma}

\subsection{Modes and Modalities as Hypergraphs}\label{ssec_NLM_in_HL}

Keeping in mind the ideas behind the definition of $\HL$, we would like to invent a way of representing formulas and sequents of the multimodal Lambek calculus $\NLM$ in $\HL$. To recall, a sequent in $\NLM$ is of the form $\Pi \to A$ where $\Pi$ is a structured database. Essentially, the latter is just a structure with brackets, which can be viewed as a term or as a tree with unary and binary branches. A natural way to transform a term into a graph is by using tree graphs; cf. e.g. the work \cite{Moot08} where the relation between TAGs, hyperedge replacement grammars and $\NLM$-grammars is established, and the constructions involve hypertrees. A hypergraph $H$ is a \emph{hypertree} if every node $v$ in $V_H$ is an attachment node of exactly two hyperedges $e_1$ and $e_2$ such that $att_H(e_1)(1) = v$ and $att_H(e_2)(i)=v$ for $i>1$, except the root node, to which exactly one hyperedge is attached (say, $e_0$) and $att_H(e_0)(1)=v$ (see \cite{Moot08}). We will implicitly use this notion later.

Note that while there are additional structural rules for modes in $\NLM$ that belong to $\A$ or to $\C$, we do not want to add new structural rules to $\HL$ but rather to use rich expressivity of hypergraph structures. Thus, if, for example, $\I = \{\alpha,\beta\}$, $\A = \{\alpha\}$, $\C = \emptyset$, then we cannot simply represent antecedents of sequents of $\NLM$ as tree graphs because this would require adding some external rules saying that the tree corresponding to the structured database $((p,q)^\alpha,r)^\alpha$ is the same as that corresponding to $(p,(q,r)^\alpha)^\alpha$. Instead of doing this, we want to make associativity of the $\alpha$ mode a part of the antecedent structure using a certain kind of hypergraphs. The difficulty is that we must also keep the $\beta$ mode non-associative. Nevertheless, if $\A = \C = \emptyset$ (there are no associative or commutative modes), then the idea of representing structured databases as hypertrees looks promising. Let us explore it.

\subsubsection{Case 1: $\A=\emptyset$, $\C = \emptyset$, $\J = \emptyset$}\label{sssec_case_1}
Our goal is to define the embedding $\mu: \mathcal{F}(\NLM) \to \mathcal{F}(\HL)$ that would relate $\NLM$ and $\HL$ in the same way as $tr$ relates $\LC$ and $\HL$. In this case, all the products are non-associative and non-commutative.
\begin{example}\label{example_NL_sequents_translations}
	Let $x,y \in \I$. Consider two simple derivable sequents of $\NLM$:
	$((p,q)^x,r)^y \to (p \bullet_x q) \bullet_y r$ and $(p,(q,r)^y)^x \to p \bullet_x (q \bullet_y r)$.	Imagine that $\mu$ is defined; let $\mu(p) = p$ for $p \in \mathit{Pr}$. If we follow the idea of representing structured databases as hypertrees, then a natural translation of these sequents into those of $\HL$ must be as follows:
	$$
	\vcenter
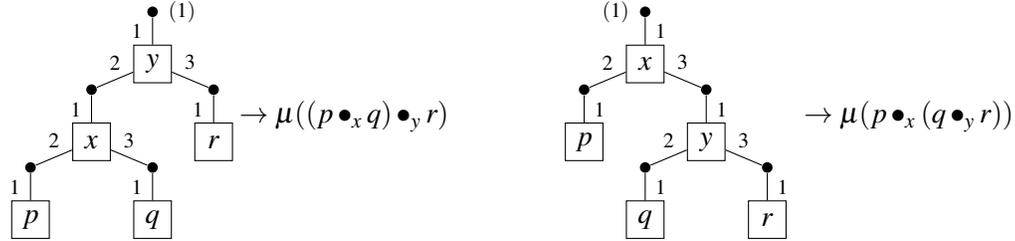
{\hbox{{\tikz[baseline=.1ex]{
					\node[node, label=right:{\scriptsize $(1)$}] (N1) {};
					\node[hyperedge, below=3.6mm of N1] (O) {$y$};
					\node[node, below left = 0.3mm and 5mm  of O] (N2) {};
					\node[node, below right = 0.3mm and 5mm  of O] (N3) {};
					\node[hyperedge, below=3.6mm of N2] (O2) {$x$};
					\node[node, below left = 0.3mm and 5mm  of O2] (N22) {};
					\node[node, below right = 0.3mm and 5mm  of O2] (N23) {};
					\node[hyperedge,below=3.6mm of N22] (E22) {$p$};
					\node[hyperedge,below=3.6mm of N23] (E23) {$q$};
					\node[hyperedge,below=3.6mm of N3] (E3) {$r$};
					\draw[-] (N1) -- node[left] {\scriptsize 1} (O);
					\draw[-] (O) -- node[above] {\scriptsize 2} (N2);
					\draw[-] (O) -- node[above] {\scriptsize 3} (N3);
					\draw[-] (N2) -- node[left] {\scriptsize 1} (O2);
					\draw[-] (O2) -- node[above] {\scriptsize 2} (N22);
					\draw[-] (O2) -- node[above] {\scriptsize 3} (N23);
					\draw[-] (N22) -- node[left] {\scriptsize 1} (E22);
					\draw[-] (N23) -- node[left] {\scriptsize 1} (E23);
					\draw[-] (N3) -- node[left] {\scriptsize 1} (E3);			
	}}}}
	\to \mu((p \bullet_x q) \bullet_y r)
	\qquad\qquad
	\vcenter{\hbox{{\tikz[baseline=.1ex]{
					\node[node, label=left:{\scriptsize $(1)$}] (N1) {};
					\node[hyperedge, below=3.6mm of N1] (O) {$x$};
					\node[node, below left = 0.3mm and 5mm  of O] (N2) {};
					\node[node, below right = 0.3mm and 5mm  of O] (N3) {};
					\node[hyperedge, below=3.6mm of N3] (O2) {$y$};
					\node[node, below left = 0.3mm and 5mm  of O2] (N22) {};
					\node[node, below right = 0.3mm and 5mm  of O2] (N23) {};
					\node[hyperedge,below=3.6mm of N22] (E22) {$q$};
					\node[hyperedge,below=3.6mm of N23] (E23) {$r$};
					\node[hyperedge,below=3.6mm of N2] (E3) {$p$};
					\draw[-] (N1) -- node[right] {\scriptsize 1} (O);
					\draw[-] (O) -- node[above] {\scriptsize 2} (N2);
					\draw[-] (O) -- node[above] {\scriptsize 3} (N3);
					\draw[-] (N3) -- node[right] {\scriptsize 1} (O2);
					\draw[-] (O2) -- node[above] {\scriptsize 2} (N22);
					\draw[-] (O2) -- node[above] {\scriptsize 3} (N23);
					\draw[-] (N22) -- node[right] {\scriptsize 1} (E22);
					\draw[-] (N23) -- node[right] {\scriptsize 1} (E23);
					\draw[-] (N2) -- node[right] {\scriptsize 1} (E3);			
	}}}}
	\to
	\mu(p \bullet_x (q \bullet_y r))
	$$
	Hyperedges of rank 3 represent modes; in particular, they are labeled by $x,y \in \I$. Since we work within the hypergraph Lambek calculus, $x,y$ must be formulas; thus, let us assume that modes are primitive formulas ($\I \subseteq \mathit{Pr}$ and $\rk(i)=3$ for $i \in \I$). In turn, formulas of $\NLM$ are translated into formulas of rank 1 that label leaves of hypertrees.
\end{example}
The above example shows us that the following hypergraph should be used to represent $p \bullet_i q$:
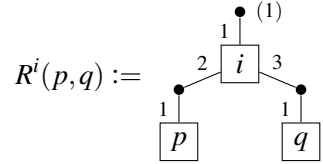
$$
R^i(p,q) \eqdef \vcenter{\hbox{{\tikz[baseline=.1ex]{
					\node[hyperedge] (O) {$i$};
					\node[node, above=3.6mm of O, label=right:{\scriptsize $(1)$}] (N1) {};
					\node[node, below left = 0.3mm and 5mm  of O] (N2) {};
					\node[node, below right = 0.3mm and 5mm  of O] (N3) {};
					\node[hyperedge,below=3.6mm of N2] (E2) {$p$};
					\node[hyperedge,below=3.6mm of N3] (E3) {$q$};
					\draw[-] (N1) -- node[left] {\scriptsize 1} (O);
					\draw[-] (N2) -- node[above] {\scriptsize 2} (O);
					\draw[-] (N3) -- node[above] {\scriptsize 3} (O);
					\draw[-] (N2) -- node[left] {\scriptsize 1} (E2);
					\draw[-] (N3) -- node[left] {\scriptsize 1} (E3);
	}}}}
$$
The hypergraph $R^i(p,q)$ is a way of composing resources $p,q$ of rank 1, so it plays the same role as the string graph $\SG(p,q)$ for $\LC$. It automatically gives rise to two residuals: $p \div R^i(q,\$_1)$ and $p \div R^i(\$_1,q)$; clearly, they correspond to $\backslash_i$ and $/_i$ respectively. Let $E_{R^i(A,B)} = \{r_0,r_1,r_2\}$ where $lab_{R^i(A,B)}(r_0) = i$, $lab_{R^i(A,B)}(r_1)=A$, $lab_{R^i(A,B)}(r_2)=B$. Note that when we consider the replacement $R^i(A,B)[r_1/H_1,r_2/H_2]$, the labels $A,B$ disappear and they do not affect the result of the replacement. In the cases when we do not care about these labels we denote $R^i(A,B)$ as $R^i$. Formally, the embedding is defined below.
\begin{definition}\label{definition_mu_1}
	Let $\I \subseteq \mathit{Pr}$, $\rk(i)=3$ for $i \in \I$. The embedding function $\mu$ is as follows:
	\begin{enumerate}
		\item $\mu(p) = p$ for $p \in \mathit{Pr}$, $\rk(p) = 1$;
		\item $\mu(A\bullet_i B) = \times\left(R^i\left(\mu(A),\mu(B)\right)\right)$;
		\item $\mu(B\backslash_i A) = \mu(A) \div R^i(\mu(B),\$_1)$;
		\item $\mu(A /_i B) = \mu(A) \div R^i(\$_1,\mu(B))$.
	\end{enumerate}
	Now, let us define the translation $\htree: \mathcal{T} \to \mathcal{H}(\mathcal{F}(\HL))$ of structured databases into hypertrees:
	\begin{enumerate}
		\item $\htree(A) = \mu(A)^\bullet$, $A \in \mathcal{F}(\NLM)$;
		\item $\htree((\Pi_1,\Pi_2)^i) = R^i[r_1/\htree(\Pi_1),r_2/\htree(\Pi_2)]$.
	\end{enumerate}
	 Finally, a sequent $\Pi \to A$ is transformed into the sequent $\htree(\Pi) \to \mu(A)$.
\end{definition}
\begin{example}
	The translation of the sequents of $\NLM$ presented in Example \ref{example_NL_sequents_translations} into those of $\HL$ according to Definition \ref{definition_mu_1} indeed gives the hypergraph sequents presented in that Example.
\end{example}

\subsubsection{Case 2: $\A=\emptyset$}\label{sssec_case_2}
The embedding $\mu$ constructed earlier can be naturally extended to capture non-associative but commutative modes $c \in \C$ as well as unary modalities $j \in \J$ by using other kinds of branches in hypergraphs. Let us do this using the following hypergraphs:
$$
K^c(A_1,A_2) \eqdef \vcenter
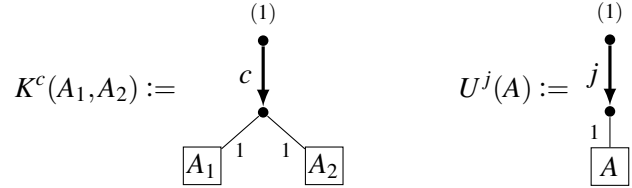
{\hbox{{\tikz[baseline=.1ex]{
				\node[node, label=above:{\scriptsize $(1)$}] (N1) {};
				\node[node, below = 8mm  of N1] (N2) {};
				\node[hyperedge,below left = 4mm and 5mm of N2] (E1) {$A_1$};
				\node[hyperedge,below right = 4mm and 5mm of N2] (E2) {$A_2$};
				\draw[very thick,-latex] (N1) -- node[left] {$c$} (N2);
				\draw[-] (N2) -- node[below] {\scriptsize 1} (E1);
				\draw[-] (N2) -- node[below] {\scriptsize 1} (E2);
}}}}
\qquad\qquad
U^j(A) \eqdef \vcenter{\hbox{{\tikz[baseline=.1ex]{
				\node[node, label=above:{\scriptsize $(1)$}] (N1) {};
				\node[node, below = 8mm  of N1] (N2) {};
				\node[hyperedge,below = 4mm of N2] (E1) {$A$};
				\draw[very thick,-latex] (N1) -- node[left] {$j$} (N2);
				\draw[-] (N2) -- node[left] {\scriptsize 1} (E1);
}}}}
$$
Let $E_{K^c(A_1,A_2)} = \{k_0,k_1,k_2\}$ where $lab_{K^c(A_1,A_2)}(k_0) = i$, $lab_{K^c(A_1,A_2)}(k_t) = A_t$ (for $t=1,2$); let also $E_{U^j(A)} = \{u_0,u_1\}$ where $lab_{U^j(A)}(u_0) = j$, $lab_{U^j(A)}(u_1) = A$. Analogously to $P^i$, let us denote $K^c(a_1,a_2)$ as $K^c$ and $U^j(a_1)$ as $U^j$ when we do not care about labels $a_1,a_2$. Note that $K^c(A_1,A_2) = K^c(A_2,A_1)$; in general, $K^c[k_1/H_1,k_2/H_2] = K^c[k_1/H_2,k_2/H_1]$ for arbitrary hypergraphs $H_1$, $H_2$ of rank 1.

For $i \in \I \setminus \C$, the inductive definition of the function $\mu$ (Definition \ref{definition_mu_1}) must not be changed. It must only be extended as follows for the cases of $c \in \C$ and $j \in \J$:
\begin{definition}\label{definition_mu_2}
	Let $\C, \J \subseteq \mathit{Pr}$ and let $\rk(c)=\rk(j)=2$ for $c \in \C$, $j \in \J$. Then for $c \in \C$ and $j \in \J$ let:
	\begin{enumerate}
		\item $\mu(A\bullet_c B) = \times\left(K^c\left(\mu(A),\mu(B)\right)\right)$;
		\item $\mu(B\backslash_c A) = \mu(A /_c B) \eqdef \mu(A) \div K^c(\mu(B),\$_1)$;
		\item $\mu(\diam_j(A)) = \times(U^j(\mu(A)))$;
		\item $\mu(\square_j(A)) = \mu(A) \div U^j(\$_1)$.
	\end{enumerate}
	The definition of $\htree: \mathcal{T} \to \mathcal{H}(\mathcal{F}(\HL))$ is extended as well:
	\begin{enumerate}
		\item $\htree((\Pi_1,\Pi_2)^c) = K^c[k_1/\htree(\Pi_1),k_2/\htree(\Pi_2)]$;
		\item $\htree(\langle \Pi \rangle^j) = U^j[u_1/\htree(\Pi)]$.
	\end{enumerate}
\end{definition}

At this point we would like to legitimize the embedding we have constructed by formulating the embedding theorem. Let us denote the set $\{\mu(A) \mid A \in \mathcal{F}(\NLM)\}$ as $\mu(\mathcal{F}(\NLM))$.
\begin{theorem}\label{th_embed_NLM}
	Let $\NLM$ be the multimodal Lambek calculus with $\A = \emptyset$.
	\begin{enumerate}
		\item If a sequent $\Pi \to T$ is derivable in $\NLM$, then $\htree(\Pi) \to \mu(T)$ is derivable in $\HL$.
		\item If $G \to X$ is a sequent over $\mu(\mathcal{F}(\NLM))$ derivable in $\HL$, then for some structured database $\Pi$ and some formula $T$ we have $G = \htree(\Pi)$, $X = \mu(T)$, and $\NLM \vdash \Pi \to T$.
	\end{enumerate}
\end{theorem}
\begin{proof}[Proof sketch]
	The first statement is proved by induction on the size of a derivation. The axiom case is as follows: $T \to T$ is translated into $\htree(T) \to \mu(T)$ equal to $\mu(T)^\bullet \to \mu(T)$, which is an axiom of $\HL$. 
	
	To prove the induction step, it suffices to observe that the translations $\htree$ and $\mu$ transform each rule of $\NLM$ into a corresponding rule of $\HL$. For example, let us check this for the rule $(\bullet_i R)$ where $i \not\in \C$:
	$$
	\infer[(\bullet_i R)]
	{
		(\Pi_1, \Pi_2)^i \to A_1 \bullet_i A_2
	}
	{
		\Pi_1 \to A_1
		&
		\Pi_2 \to A_2
	}
	\qquad
	\rightsquigarrow
	\qquad
	\infer[]
	{
		\htree((\Pi_1, \Pi_2)^i) \to \mu(A_1 \bullet_i A_2)
	}
	{
		\htree(\Pi_1) \to \mu(A_1)
		&
		\htree(\Pi_2) \to \mu(A_2)
	}
	$$
	Note that
	\begin{itemize}
		\item $\mu(A_1 \bullet_i A_2) = \times(R^i(\mu(A_1),\mu(A_2)))$;
		\item $\htree((\Pi_1,\Pi_2)^i) = R^i(\mu(A_1),\mu(A_2))[r_1/\htree(\Pi_1),r_2/\htree(\Pi_2)]$.
	\end{itemize} 
	Hence the sequent $\htree((\Pi_1, \Pi_2)^i) \to \mu(A_1 \bullet_i A_2)$ is indeed obtained from the premises $\htree(\Pi_1) \to \mu(A_1)$ and $\htree(\Pi_2) \to \mu(A_2)$ according to the rule $(\times R)$ (consult Definition \ref{definition_HL_rules}).
	
	Let us do the same with the rule $(\square_j L)$. Given the rule as in Section \ref{sec_introduction}, note that there is a hyperedge $e$ of rank 1 in $\htree(\Delta[A])$ labeled by $N = \mu(A)$ that corresponds to the distinguished occurrence of $A$ in $\Delta[A]$. By the definition of $\mu$, $\mu(\square_j A) = \mu(A) \div U^j(\$_1) = N \div D$ where $D = U^j(\$_1)$, $E_D = \{u_0,u_1\}$, $lab_D(u_0) = j$, $lab_D(u_1)=\$_1$ (note that $e^\$_D = u_1$). Then the following is a correct application of $(\div L)$:
	$$
	\infer[(\div L)]
	{
		\htree(\Delta[A])\left[e/D[u_1/(N \div D)^\bullet, u_0/j^\bullet]\right] \to \mu(C)
	}
	{
		\htree(\Delta[A]) \to \mu(C)
		&
		\qquad
		&
		j^\bullet \to j
	}
	$$
	Finally, 
	$
		\htree(\Delta[A])\left[e/D[u_1/(N \div D)^\bullet, u_0/j^\bullet]\right] = 
		\htree(\Delta[A])\left[e/U^j(\mu(\square_j A))]\right]  =
		\htree(\Delta[\langle \square_j A \rangle^j])
	$. 
	Thus the application of the rule $(\square_j L)$ is translated into a correct application of the rule $(\div R)$.
	
	Let us also check the rule $(\square_j R)$. It suffices to notice that $\htree(\langle  \Pi \rangle^j) = U^j(\$_1)[u_1/\htree(\Pi)]$; this implies that the following is a correct rule application of $(\div R)$:
	$$
	\infer[(\div R)]
	{
		\htree(\Pi) \to \mu(\square_j C)
	}
	{
		\htree(\langle  \Pi \rangle^j) \to \mu(C)
	}
	$$
	The same can be done with every rule of $\NLM$. Finally, note that if we translate the premise and the conclusion in the rule $(P)$ for $c \in \C$ into sequents of $\HL$, then the premise coincides with the conclusion since $\htree((\Pi_1,\Pi_2)^c) = K^c[k_1/\htree(\Pi_1),k_2/\htree(\Pi_2)] = \htree((\Pi_2,\Pi_1)^c)$. Summing up, given a derivation of $\Pi \to A$ in $\NLM$, one can translate each sequent $\Delta \to B$ in it into $\htree(\Delta) \to \mu(B)$ and obtain a correct derivation in $\HL$ with the rule applications of $(P)$ becoming tautological.
	
	To prove the second statement, firstly, note that $X = \mu(T)$ since $G \to X$ is over $\mu(\mathcal{F}(\NLM))$. Now, let us prove the statement by induction on the length of a derivation of $G \to X$. The induction base is when $G \to X$ is an axiom; then $G = \mu(T)^\bullet = \htree(X)$, q.e.d. To prove the induction step let us consider the last rule applied in the derivation of $G \to X$.
	
	\textit{Case $(\div L)$.} Let the last rule applied be $(\div L)$. Assume that the major formula in this rule is $N \div D = \mu(A/_i B) = \mu(A) \div R^i(\$_1,\mu(B))$ (here $i \not\in\C$). Then the rule application must be of the form
	$$
	\infer[(\div L)]
	{
		H\left[e/R^i(\$_1,\mu(B))[r_1/(\mu(A /_i B))^\bullet, r_0/H_0, r_2/H_2]\right] 
		\to \mu(T)
	}{
		H \to \mu(T)
		&
		H_0 \to i
		&
		H_2 \to \mu(B)
	}
	$$
	Here $e \in E_H$ and $lab_H(e) = \mu(A)$. By the induction hypothesis, $H = \htree(\Psi)$, $H_2 = \htree(\Pi)$ and \\$\NLM \vdash \Psi \to T$, $\NLM \vdash \Pi \to B$. Note that $\Psi$ has an occurrence of $A$ corresponding to the formula $\mu(A)$ labeling $e$; let us distinguish this occurrence by denoting $\Psi$ as $\Delta[A]$. 
	
	The problem is in the sequent $H_0 \to i$: here the induction hypothesis cannot be applied since $i$ is not the image of $\mu$ but is a primitive formula representing the mode $i$. Instead, let us apply Lemma \ref{lemma_wolf}. We can do this since any subformula of a formula from $\mu(\mathcal{F}(\NLM))$ either is of the form $\mu(C)$ or it belongs to $\I \cup \J$; these formulas are skeleton-free and the heads of formulas from $\mu(\mathcal{F}(\NLM))$ cannot contain $i$ alone (only along with some primitive formulas of rank 1); thus $\head(B)=i$ implies $B=i$ for $B$ of interest. Consequently, according to Lemma \ref{lemma_wolf} $H_0 = i^\bullet$. Finally, 
	\begin{eqnarray*}
		H\left[e/R^i(\$_1,\mu(B))[r_1/(\mu(A /_i B))^\bullet, r_0/H_0, r_2/H_2]\right]
		= 
		H\left[e/R^i(\mu(A /_i B),\mu(B))[r_2/\htree(\Pi)]\right]
		=\\=
		\htree(\Delta[(A/_i B,\Pi)^i])
	\end{eqnarray*}
	The cases where $N \div D$ equals $\mu(B \backslash_i A)$, or $\mu(A/_c B)$ for $c \in \C$, or $\mu(\square_j A)$ are dealt with similarly.
	
	\textit{Case $(\div R)$.} Let the last rule applied be $(\div R)$. Assume that the major formula is $X = N \div D = \mu(T) = \mu(A/_i B) = \mu(A) \div R^i(\$_1,\mu(B))$ (where $i \not\in\C$). Then the rule application must be of the form
	$$
	\infer[(\div R)]
	{
		F\to \mu(A) \div R^i(\$_1,\mu(B))
	}
	{
		R^i(\$_1,\mu(B))[r_1/F] \to \mu(A)
	}
	$$
	By the induction hypothesis, $R^i(\$_1,\mu(B))[r_1/F] = \htree(\Psi)$ and $\NLM \vdash \Psi \to A$. The hypergraph $H = R^i(\$_1,\mu(B))[r_1/F]$ looks as follows: it has one external node, which is also the first attachment node of an $i$-labeled hyperedge, which has two daughters; the rightmost one goes to the leaf labeled by $\mu(B)$, while the leftmost one goes to the subhypergraph $F$ that replaces $r_1$. We know that the whole hypergraph $H$ is of the form $\htree(\Psi)$; then the subhypergraph $F$ must also be of this form according to the inductive definition of $\htree$. Hence $F = \htree(\Pi)$, and $H = \htree((\Pi,B)^i)$. Finally, note that $\Pi \to A/_i B$ is obtained from $(\Pi,B)^i \to A$ according to $(/_i R)$. 
	
	The remaining cases for the rules $(\div R)$, $(\times L)$ and $(\times R)$ are dealt with similarly.
\end{proof}
The proof of Theorem \ref{th_embed_NLM} shows us that the rules of $\NLM$ are translated into the instances of the rules of $\HL$, and, conversely, if there is a rule application of $\HL$ with the conclusion of the form $\htree(\Pi) \to \mu(T)$, then it must be a translation of a rule application of $\NLM$. Therefore, the embedding of $\NLM$ in $\HL$ can be considered as strong. Note, however, that in $\HL$ additional axioms of the form $i^\bullet \to i$ appear for $i \in \I \cup \J$; they are required since modes and modalities are treated as primitive formulas in $\HL$. 

\begin{example}
	Below several examples of translations of formulas and structured databases according to $\mu$ are presented (where $c \in \C$; $i, x, y \in \I \setminus \C$; $j \in \J$):
	\begin{center}
		\begin{tabular}{|c|c|c|}
		\hline
		$
		\mu(\square_j\diam_j p) = 
		$
		&
		$
		\mu((q \backslash_x p)/_y r) = 
		$
		&
		$
		\htree\left( ((p,q)^c,\langle r \rangle^j)^i \right) =
		$
		\\
		\hline
		$
		\times
		\left(
		\vcenter{\hbox{{\tikz[baseline=.1ex]{
						\node[node, label=above:{\scriptsize $(1)$}] (N1) {};
						\node[node, below = 8mm  of N1] (N2) {};
						\node[hyperedge,below = 4mm of N2] (E1) {$p$};
						\draw[very thick,-latex] (N1) -- node[left] {$j$} (N2);
						\draw[-] (N2) -- node[left] {\scriptsize 1} (E1);
		}}}}
		\right)
		\div
		\left(
		\vcenter{\hbox{{\tikz[baseline=.1ex]{
						\node[node, label=above:{\scriptsize $(1)$}] (N1) {};
						\node[node, below = 8mm  of N1] (N2) {};
						\node[hyperedge,below = 4mm of N2] (E1) {$\$_1$};
						\draw[very thick,-latex] (N1) -- node[left] {$j$} (N2);
						\draw[-] (N2) -- node[left] {\scriptsize 1} (E1);
		}}}}
		\right)
		$
		&
		$
		\left(
		p \div 
		\left(
		\vcenter{\hbox{{\tikz[baseline=.1ex]{
						\node[hyperedge] (O) {$x$};
						\node[node, above=3.6mm of O, label=above:{\scriptsize $(1)$}] (N1) {};
						\node[node, below left = 0.3mm and 3mm  of O] (N2) {};
						\node[node, below right = 0.3mm and 3mm  of O] (N3) {};
						\node[hyperedge,below=3.6mm of N2] (E2) {$q$};
						\node[hyperedge,below=3.6mm of N3] (E3) {$\$_1$};
						\draw[-] (N1) -- node[left] {\scriptsize 1} (O);
						\draw[-] (N2) -- node[above] {\scriptsize 2} (O);
						\draw[-] (N3) -- node[above] {\scriptsize 3} (O);
						\draw[-] (N2) -- node[left] {\scriptsize 1} (E2);
						\draw[-] (N3) -- node[left] {\scriptsize 1} (E3);
		}}}}
		\right)
		\right)
		\div 
		\left(
		\vcenter{\hbox{{\tikz[baseline=.1ex]{
						\node[hyperedge] (O) {$y$};
						\node[node, above=3.6mm of O, label=above:{\scriptsize $(1)$}] (N1) {};
						\node[node, below left = 0.3mm and 3mm  of O] (N2) {};
						\node[node, below right = 0.3mm and 3mm  of O] (N3) {};
						\node[hyperedge,below=3.6mm of N2] (E2) {$\$_1$};
						\node[hyperedge,below=3.6mm of N3] (E3) {$r$};
						\draw[-] (N1) -- node[left] {\scriptsize 1} (O);
						\draw[-] (N2) -- node[above] {\scriptsize 2} (O);
						\draw[-] (N3) -- node[above] {\scriptsize 3} (O);
						\draw[-] (N2) -- node[left] {\scriptsize 1} (E2);
						\draw[-] (N3) -- node[left] {\scriptsize 1} (E3);
		}}}}
		\right)
		$
		&
		$
		\vcenter{\hbox{{\tikz[baseline=.1ex]{
						\node[node, label=right:{\scriptsize $(1)$}] (N1) {};
						\node[hyperedge, below=3.6mm of N1] (O) {$i$};
						\node[node, below left = 0.3mm and 5mm  of O] (N2) {};
						\node[node, below right = 0.3mm and 5mm  of O] (N3) {};
						\node[node, below = 6.6mm of N2] (N22) {};
						\node[node, below = 6.6mm of N3] (N31) {};
						\node[hyperedge,below=3.6mm of N31] (E31) {$r$};
						\node[hyperedge,below left = 4mm and 3mm  of N22] (E22) {$p$};
						\node[hyperedge,below right = 4mm and 3mm  of N22] (E23) {$q$};
						\draw[very thick,-latex] (N2) -- node[left] {$c$} (N22);
						\draw[very thick,-latex] (N3) -- node[left] {$j$} (N31);
						\draw[-] (N1) -- node[left] {\scriptsize 1} (O);
						\draw[-] (O) -- node[above] {\scriptsize 2} (N2);
						\draw[-] (O) -- node[above] {\scriptsize 3} (N3);
						\draw[-] (N22) -- node[left] {\scriptsize 1} (E22);
						\draw[-] (N22) -- node[right] {\scriptsize 1} (E23);
						\draw[-] (N31) -- node[left] {\scriptsize 1} (E31);
		}}}}
		$
		\\[-11pt]
		&&\\
		\hline
		\end{tabular}
	\end{center}
\end{example}

\begin{example}
	Below the derivations of the sequents $p \to \square_j \diam_j p$ and $q \bullet_c (p \mathbin{/_c} q) \to p$ are presented in the first row (where $j \in \J$, $c \in \C$). In the second row, the derivations of their translations in $\HL$ is presented.
	\begin{center}
	\begin{tabular}{|c|c|}
	\hline
	&\\[-9pt]
	$
	\infer[(\square_j R)]
	{
		p \to \square_j \diam_j p
	}
	{
		\infer[(\diam_j R)]
		{
			\langle p \rangle^j \to \diam_j p
		}
		{
			p \to p
		}
	}
	$
	&
	$
	\infer[(\bullet_c L)]
	{
		q \bullet_c (p\mathbin{/_c} q) \to p
	}
	{
		\infer[(P)]
		{
			(q,p\mathbin{/_c} q)^c \to p
		}
		{
			\infer[(\mathop{/_c} L)]
			{
				(p\mathbin{/_c} q, q)^c \to p
			}
			{
				p \to p
				&
				q \to q
			}
		}
	}
	$
	\\[3pt]
	\hline
	&\\[-9pt]
	$\qquad
	\infer[(\div R)]
	{
		\vcenter{\hbox{{\tikz[baseline=.1ex]{
						\node[node, label=above:{\scriptsize $(1)$}] (N1) {};
						\node[hyperedge, below=4mm of N1] (O) {$p$};
						\draw[-] (N1) -- node[left] {\scriptsize 1} (O);
		}}}}
		\to \mu(\square_j \diam_j p)
	}
	{
		\infer[(\times R)]
		{
			\vcenter{\hbox{{\tikz[baseline=.1ex]{
							\node[node, label=above:{\scriptsize $(1)$}] (N1) {};
							\node[node, below = 7mm  of N1] (N2) {};
							\node[hyperedge,below = 4mm of N2] (E1) {$p$};
							\draw[very thick,-latex] (N1) -- node[left] {$j$} (N2);
							\draw[-] (N2) -- node[left] {\scriptsize 1} (E1);
			}}}}
			\to
			\times\left(
			\vcenter{\hbox{{\tikz[baseline=.1ex]{
							\node[node, label=above:{\scriptsize $(1)$}] (N1) {};
							\node[node, below = 7mm  of N1] (N2) {};
							\node[hyperedge,below = 4mm of N2] (E1) {$p$};
							\draw[very thick,-latex] (N1) -- node[left] {$j$} (N2);
							\draw[-] (N2) -- node[left] {\scriptsize 1} (E1);
			}}}}
			\right)
		}
		{
			p^\bullet \to p
			&
			j^\bullet \to j
		}
	}
	\qquad
	$
	&
	$\qquad
	\infer[(\times L)]
	{
		\vcenter{\hbox{{\tikz[baseline=.1ex]{
						\node[node, label=above:{\scriptsize $(1)$}] (N1) {};
						\node[hyperedge, below=4mm of N1] (O) {$\mu(T)$};
						\draw[-] (N1) -- node[left] {\scriptsize 1} (O);
		}}}}
		\to p
	}
	{
		\infer[(\div L)]
		{
			\vcenter{\hbox{{\tikz[baseline=.1ex]{
							\node[node, label=above:{\scriptsize $(1)$}] (N1) {};
							\node[node, below = 7mm  of N1] (N2) {};
							\node[hyperedge,below left = 4mm and 5mm of N2] (E1) {$q$};
							\node[hyperedge,below right = 4mm and 2mm of N2] (E2) {$\,\mu(p\,/\!_c\, q)\,$};
							\draw[very thick,-latex] (N1) -- node[left] {$c$} (N2);
							\draw[-] (N2) -- node[left] {\scriptsize 1} (E1);
							\draw[-] (N2) -- node[right] {\scriptsize 1} (E2);
			}}}}
			\to p
		}
		{
			p^\bullet \to p
			&
			q^\bullet \to q
			&
			c^\bullet \to c
		}
	}
	\qquad
	$
	\\[-9pt]
	&\\
	\hline
	\end{tabular}
	\end{center}
\end{example}

\subsubsection{Case 3: $\A = \{O\}$, $O \in \C$}\label{sssec_case_3}
Now, consider $\NLM$ with one associative mode $O$, which is also commutative (there can be arbitrary many non-associative modes and modalities as well). Let $K^O(A_1,A_2)$ be the following hypergraph:
$$
K^O(A_1,A_2) \eqdef 
\vcenter
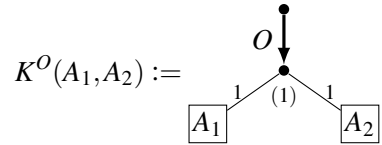
{\hbox{{\tikz[baseline=.1ex]{
				\node[node] (N1) {};
				\node[node, below = 6.6mm  of N1, label=below:{\scriptsize $(1)$}] (N2) {};
				\node[hyperedge,below left = 4mm and 7mm of N2] (E1) {$A_1$};
				\node[hyperedge,below right = 4mm and 7mm of N2] (E2) {$A_2$};
				\draw[very thick,-latex] (N1) -- node[left] {$O$} (N2);
				\draw[-] (N2) -- node[left] {\scriptsize 1} (E1);
				\draw[-] (N2) -- node[right] {\scriptsize 1} (E2);
}}}}
$$
Let us extend the inductive definition of $\mu$ to the case of the new mode as follows:
\begin{enumerate}
	\item $\mu(A\bullet_O B) = \times\left(K^O\left(\mu(A),\mu(B)\right)\right)$;
	\item $\mu(B\backslash_O A) = \mu(A /_O B) = \mu(A) \div K^O(\$_1,\mu(B))$.
\end{enumerate}

The inductive definition of $\mu$ is not changed for the remaining modalities as well as for modes. The definition of $\htree$ is also extended in a usual way. Correctness of the translation can be proved in the same way as in the previous cases.

\subsubsection{Case 4: $\A = \{O\}$, $O \not\in \C$}\label{sssec_case_4}
This is the trickiest case. Apparently, the way of composition of two hypergraphs $H_1$, $H_2$ in an associative but a non-commutative way can be done only in a string-like way, so we have to have ``the beginning'' and ``the end'' of $H_1$ and $H_2$ in order to connect ``the end'' of $H_1$ to ``the beginning'' of $H_2$. This implies that $H_1$, $H_2$ should have rank 2 rather than rank 1. Thus we have to change the definition of $\mu$ in all the cases. Let us define the following hypergraphs (where $O \in \A$, $c \in \C$, $i \in \I \setminus (\A \cup \C)$, $j \in \J$):
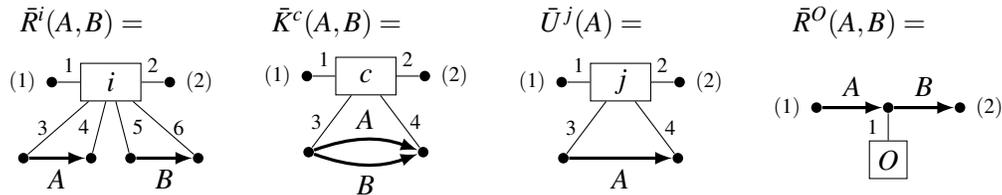
\begin{center}
	\begin{tabular}{cccc}
		$
		\bar{R}^i(A,B) = \qquad
		$
		&
		$
		\bar{K}^c(A,B) = \qquad
		$
		&
		$
		\bar{U}^j(A) = \qquad
		$
		&
		$
		\bar{R}^O(A,B) = \qquad
		$
		\\
		$
		\vcenter{\hbox{{\tikz[baseline=.1ex]{
						\node[hyperedgewide] (O) {$i$};
						\node[node, left=3mm of O, label=left:{\scriptsize $(1)$}] (N1l) {};
						\node[node, right=3mm of O, label=right:{\scriptsize $(2)$}] (N1r) {};
						\node[node, below left = 7mm and 7mm  of O] (N2l) {};
						\node[node, below left = 7mm and -2mm  of O] (N2r) {};
						\node[node, below right = 7mm and -2mm  of O] (N3l) {};
						\node[node, below right = 7mm and 7mm  of O] (N3r) {};
						\draw[very thick,-latex] (N2l) -- node[below] {$A$} (N2r);
						\draw[very thick,-latex] (N3l) -- node[below] {$B$} (N3r);
						\draw[-] (N1l) -- node[above] {\scriptsize 1} (O);
						\draw[-] (N1r) -- node[above] {\scriptsize 2} (O);
						\draw[-] (O) -- node[left] {\scriptsize 3} (N2l);
						\draw[-] (O) -- node[left] {\scriptsize 4} (N2r);
						\draw[-] (O) -- node[right] {\scriptsize 5} (N3l);
						\draw[-] (O) -- node[right] {\scriptsize 6} (N3r);
		}}}}
		$
		&
		$
		\vcenter{\hbox{{\tikz[baseline=.1ex]{
						\node[hyperedgewide] (O) {$c$};
						\node[node, left=3mm of O, label=left:{\scriptsize $(1)$}] (N1l) {};
						\node[node, right=3mm of O, label=right:{\scriptsize $(2)$}] (N1r) {};
						\node[node, below left = 7mm and 3mm  of O] (N2l) {};
						\node[node, below right = 7mm and 3mm  of O] (N2r) {};
						\draw[very thick,-latex] (N2l) to[bend left = 20] node[above] {$A$} (N2r);
						\draw[very thick,-latex] (N2l) to[bend right = 20] node[below] {$B$} (N2r);
						\draw[-] (N1l) -- node[above] {\scriptsize 1} (O);
						\draw[-] (N1r) -- node[above] {\scriptsize 2} (O);
						\draw[-] (O) -- node[left] {\scriptsize 3} (N2l);
						\draw[-] (O) -- node[right] {\scriptsize 4} (N2r);
		}}}}
		$
		&
		$
		\vcenter{\hbox{{\tikz[baseline=.1ex]{
						\node[hyperedgewide] (O) {$j$};
						\node[node, left=3mm of O, label=left:{\scriptsize $(1)$}] (N1l) {};
						\node[node, right=3mm of O, label=right:{\scriptsize $(2)$}] (N1r) {};
						\node[node, below left = 7mm and 3mm  of O] (N2l) {};
						\node[node, below right = 7mm and 3mm  of O] (N2r) {};
						\draw[very thick,-latex] (N2l) -- node[below] {$A$} (N2r);
						\draw[-] (N1l) -- node[above] {\scriptsize 1} (O);
						\draw[-] (N1r) -- node[above] {\scriptsize 2} (O);
						\draw[-] (O) -- node[left] {\scriptsize 3} (N2l);
						\draw[-] (O) -- node[right] {\scriptsize 4} (N2r);
		}}}}
		$
		&
		$
		\vcenter{\hbox{{\tikz[baseline=.1ex]{
						\node[node, label=left:{\scriptsize $(1)$}] (N1) {};
						\node[node, right=8mm of N1] (N2) {};
						\node[node, right=8mm of N2, label=right:{\scriptsize $(2)$}] (N3) {};
						\node[hyperedge,below = 3.6mm of N2] (E) {$O$};
						\draw[-] (N2) -- node[left] {\scriptsize 1} (E);
						\draw[very thick,-latex] (N1) -- node[above] {$A$} (N2);
						\draw[very thick,-latex] (N2) -- node[above] {$B$} (N3);
		}}}}
		$
		\\
	\end{tabular}
\end{center}
The new embedding functions $\bar{\mu}$ and $\htreebar$ are defined as in the previous cases but with new hypergraphs:
\begin{itemize}
	\item $\bar{\mu}(p) \eqdef p$ for $p \in \mathit{Pr}$ where $\rk(p) = 2$;
	\item $\bar{\mu}(A \bullet_i B) = \times\left(\bar{R}^i\left(\bar{\mu}(A),\bar{\mu}(B)\right)\right)$ ($i \in \I \setminus \C$); $\bar{\mu}(A \bullet_c B) = \times\left(\bar{K}^c\left(\bar{\mu}(A),\bar{\mu}(B)\right)\right)$ ($c \in \C$); etc.
	\item $\htreebar(A) = A^\bullet$; $\htreebar((\Pi_1,\Pi_2)^i) = \bar{R}^i(A,B)[\bar{r}_1/\htreebar(\Pi_1),\bar{r}_2/\htreebar(\Pi_2)]$; etc.
\end{itemize}
\begin{example}
	If $\Pi = (p,(\langle q\rangle^j,(r,(s,t)^O)^c)^O)^O$ and $A$ is a formula, then $\htreebar(\Pi) \to \bar{\mu}(A)$ is of the form:
	$$
	\vcenter
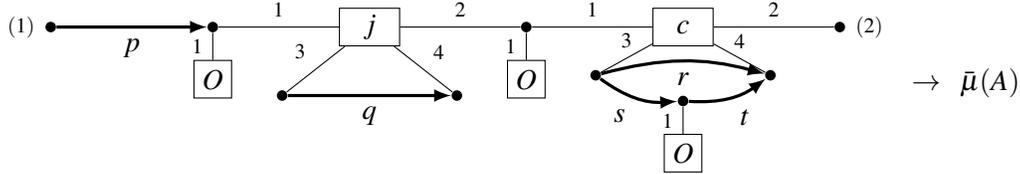
{\hbox{{\tikz[baseline=.1ex]{
					\node[hyperedgewide] (O) {$j$};
					\node[node, left=16mm of O] (N1l) {};
					\node[node, left=20mm of N1l, label=left:{\scriptsize $(1)$}] (N-1) {};
					\node[node, right=16mm of O] (N1r) {};
					\node[node, below left = 6mm and 7mm  of O] (N2l) {};
					\node[node, below right = 6mm and 7mm  of O] (N2r) {};
					\node[hyperedgewide, right=16mm of N1r] (Op) {$c$};
					\node[node, right=16mm of Op, label=right:{\scriptsize $(2)$}] (Np1r) {};
					\node[node, below left = 3.3mm and 7mm  of Op] (Np2l) {};
					\node[node, below right = 3.3mm and 7mm  of Op] (Np2r) {};
					\node[node, below = 6.5mm of Op] (Np3) {};
					\node[hyperedge,below = 3.6mm of N1l] (Ass1) {$O$};
					\node[hyperedge,below = 3.6mm of N1r] (Ass2) {$O$};
					\node[hyperedge,below = 3.6mm of Np3] (Ass3) {$O$};
					\draw[-] (N1l) -- node[left] {\scriptsize 1} (Ass1);
					\draw[-] (N1r) -- node[left] {\scriptsize 1} (Ass2);
					\draw[-] (Np3) -- node[left] {\scriptsize 1} (Ass3);
					\draw[very thick,-latex] (Np2l) to[bend left = 15] node[below] {$r$} (Np2r);
					\draw[very thick,-latex] (Np2l) to[bend right = 20] node[below left] {$s$} (Np3);
					\draw[very thick,-latex] (Np3) to[bend right = 20] node[below right] {$t$} (Np2r);
					\draw[-] (N1r) -- node[above] {\scriptsize 1} (Op);
					\draw[-] (Np1r) -- node[above] {\scriptsize 2} (Op);
					\draw[-] (Op) -- node[above] {\scriptsize 3} (Np2l);
					\draw[-] (Op) -- node[above] {\scriptsize 4} (Np2r);
					\draw[very thick,-latex] (N-1) -- node[below] {$p$} (N1l);
					\draw[very thick,-latex] (N2l) -- node[below] {$q$} (N2r);
					\draw[-] (N1l) -- node[above] {\scriptsize 1} (O);
					\draw[-] (N1r) -- node[above] {\scriptsize 2} (O);
					\draw[-] (O) -- node[above left] {\scriptsize 3} (N2l);
					\draw[-] (O) -- node[above right] {\scriptsize 4} (N2r);
	}}}}
	\;\;\to\;\; \bar{\mu}(A)
	$$
	The antecedent combines hypertree-like parts with string-like parts connected in series or in parallel (this reminds us electrical networks). This is what we wanted to reach from the beginning: to completely avoid adding structural rules for different modes but to make them a part of sequent structure. In the end, we obtain a nice visual form of sequents of $\NLM$ highlighting their associative/commutative parts.
\end{example}
Finally, the theorem similar to Theorem \ref{th_embed_NLM} but where $\mu$ is replaced by $\bar{\mu}$ can be proved, so the functions $\htreebar$ and $\bar{\mu}$ embed $\NLM$ in $\HL$. As a corollary, the Lambek calculus with brackets \cite{Morrill92}, for which $|\I|=|\A| = 1$, $\C = \emptyset$, $|\J| = 1$, can be embedded in $\HL$ as well using $\bar{\mu}$.

\begin{remark}
	If $|\A|>1$, then it is unclear whether we can define an embedding of $\NLM$ in $\HL$. Note that if $\A = \{O_1,O_2\}$ and if we use the hypergraph $\bar{R}^O(A,B)$ to represent these modes (for $O \in \{O_1,O_2\}$), then they start interacting with each other in an undesirable way: $\htreebar(((p,q)^{O_1},r)^{O_2}) = \htreebar((p,(q,r)^{O_2})^{O_1})$.
\end{remark}

\section{Conclusion}\label{sec_discussion_conclusion}

The hypergraph Lambek calculus is a general framework with cumbersome but visually informative syntax using hyperedge replacement. The embedding of $\NLM$ in $\HL$ presented in this work incorporates the structural postulates for modes in the sequent structure, so they become properties of hypergraphs rather than external rules. It should be noted that Cases 1-4 considered above imply that all the calculi $\NL$, $\mathrm{NLP}$, $\LC$, $\LP$, and the Lambek calculus with brackets \cite{Morrill92} can be embedded in $\HL$.

The results of this paper along with those of \cite{Pshenitsyn22} justify the role of $\HL$ as of a very general pure logic of residuation, maybe the most general one (since hypergraphs may be viewed as the most general discrete finite structures). In that respect, $\HL$ competes with $\MILLFO$, in which many extensions of $\LC$ can be embedded as well. We claim that $\HL$ is a better candidate to the role of such an ``umbrella logic'' because the embeddings presented in \cite{Pshenitsyn22} and in this paper are \emph{strong}: each derivation in a logic ($\LC$, $\LP$ etc.) can be translated rule-by-rule into that in the hypergraph Lambek calculus, and vise versa. The embeddings in $\MILLFO$ are not strong in this sense, since a proof in $\MILLFO$ must be reorganized before translating it back into a proof in a logic. Moreover, we have shown existence of a straightforward and natural translation of $\NLM$ in $\HL$ using hypertrees --- while no such embedding of $\NLM$ in $\MILLFO$ have been found. 


\bibliographystyle{eptcs}
\bibliography{AMSLO23}

\end{document}